\documentclass{article}
\usepackage{ijcai20}
\usepackage{times}
\usepackage{soul}
\usepackage{url}
\usepackage[hidelinks]{hyperref}
\usepackage[utf8]{inputenc}
\usepackage[small]{caption}
\usepackage{graphicx}
\usepackage{amsmath}
\usepackage{amsthm}
\usepackage{booktabs}
\usepackage{algorithm}
\usepackage{algorithmic}
\usepackage{color}
\usepackage{mathrsfs}
\usepackage{amssymb,verbatim}
\usepackage{amsfonts}
\usepackage{booktabs,microtype}
\usepackage{tikz}
\usetikzlibrary{arrows,shapes, calc, fit, positioning}
\usepackage{tkz-graph}
\newcommand{\bbR}{\mathbb{R}}
\newcommand{\bbZ}{\mathbb{Z}}
\newcommand{\bbN}{\mathbb{N}}

\newcommand{\calA}{\mathcal{A}}

\newcommand{\calC}{\mathcal{C}}
\newcommand{\calD}{\mathcal{D}}

\newcommand{\calI}{\mathcal{I}}

\newcommand{\calX}{\mathcal{X}}

\newcommand{\msucceq}{\mathop{\stackrel{m}{\succeq}}}

\newcommand{\ceil}[1]{\lceil #1 \rceil }

\usepackage[textsize=footnotesize,color=red!25,bordercolor=red]{todonotes}

\usepackage[round]{natbib}
\usepackage{graphicx}
\usepackage{xcolor}
\usepackage{hyperref}
\hypersetup{
     colorlinks   = true,
     linkcolor    = red, 
     urlcolor     = blue, 
	 citecolor    = blue 
}

\usepackage{booktabs}
\usepackage{pifont}
\newcommand{\cmark}{\ding{51}}%
\newtheorem{example}{Example}
\newtheorem{theorem}{Theorem}
\newtheorem{lemma}{Lemma}
\newtheorem{corollary}{Corollary}

\usepackage{etoolbox}
\newcounter{Bew1}
\newcounter{Bew2}

\title{Fair Division of Time: Multi-layered Cake Cutting}

\author{
Hadi Hosseini$^1$\and
Ayumi Igarashi$^2$\and
Andrew Searns$^1$
\affiliations
$^1$Rochester Institute of Technology, US\\
National Institute of Informatics, Japan
\emails
hhvcs@rit.edu,
ayumi\_igarashi@nii.ac.jp,
abs2157@rit.edu
}

\begin{document}

\maketitle

\begin{abstract}
We initiate the study of multi-layered cake cutting with the goal of fairly allocating multiple divisible resources (layers of a cake) among a set of agents. The key requirement is that each agent can only utilize a single resource at each time interval. Several real-life applications exhibit such restrictions on overlapping pieces; for example, assigning time intervals over multiple facilities and resources or assigning shifts to medical professionals. We investigate the existence and computation of envy-free and proportional allocations. We show that envy-free allocations that are both feasible and contiguous are guaranteed to exist for up to three agents with two types of preferences, when the number of layers is two. We also show that envy-free feasible allocations where each agent receives a polynomially bounded number of intervals exist for any number of agents and layers under mild conditions on agents' preferences. We further devise an algorithm for computing proportional allocations for any number of agents and layers. 
\end{abstract}

\section{Introduction}
Consider a group of students who wish to use multiple college facilities such as a conference room and an exercise room over different periods of time.
Each student has a preference over what facility to use at different time of the day: Alice prefers to set her meetings in the morning and exercise in the afternoon, whereas Bob prefers to start the day with exercising for a couple of hours and meet with his teammates in the conference room for the rest of the day.

The fair division literature has extensively studied the problem of dividing a heterogeneous divisible resource (aka a \textit{cake}) among several agents who may have different preference over the various pieces of the cake~\citep{Steinhaus48,Robertson98,Brams06}.
These studies have resulted in a plethora of axiomatic and existence results~\citep{barbanel2005geometry,moulin2004fair} as well as computational solutions~\citep{procaccia2013cake,aziz2016discretefour} under a variety of assumptions, and were successfully implemented in practice (see~\citep{procaccia_moulin_2016,Brams96} for an overview).
In the case of Alice and Bob, each facility represents a layer of the cake in a \textit{multi-layered cake cutting} problem, and the question is how to allocate the time intervals (usage right) of the facilities according to their preferences in a fair manner.

One naive approach is to treat each cake independently and solve the problem through well-established cake-cutting techniques by performing a fair division on each layer separately.
However, this approach has major drawbacks: First, the final outcome, although fair on each layer, may not necessarily be fair overall. Second, the allocation may not be feasible, i.e., it may assign two overlapping pieces (time intervals) to a single agent. 
In our example, Alice cannot simultaneously utilize the exercise room and the conference room at the same time if she receives overlapping  intervals.
Several other application domains exhibit similar structures over resources: assigning nurses to various wards and shifts, doctors to operation rooms, and research equipment to groups, to name a few.

In multi-layared cake cutting, each layer represents a divisible resource. Each agent has additive preferences over every disjoint (non-overlapping) intervals. A division of a multi-layered cake is \emph{feasible} if no agent's share contains overlapping intervals, and is contiguous if each allocated piece of a layer is contiguous. 
There has been some recent work on dividing multiple cakes among agents~\citep{cloutier2010two,lebert2013envy}. Yet, none of the previous work considered the division of multiple resources under feasibility and contiguity constraints. 
In this paper, we thus ask the following question: 
\begin{quote}
\textit{What fairness guarantees can be achieved under feasibility and contiguity constraints for various number of agents and layers?}
\end{quote}

\subsection{Our Results}
We initiate the study of the multi-layered cake cutting problem for allocating divisible resources, under contiguity and feasibility requirements. Our focus is on two fairness notions, \textit{envy-freeness} and \textit{proportionality}. Envy-freeness (EF) requires that each agent believes no other agent's share is better than its share of the cake. Proportionality (Prop) among $n$ agents requires that each agent receives a share that is valued at least $\frac{1}{n}$ of the value of the entire cake.
For efficiency, we consider \textit{complete} divisions with no leftover pieces.

Focusing on envy-free divisions, we show the existence of envy-free and complete allocations that are both feasible and contiguous for two-layered cakes and up to three agents with at most two types of preferences. These cases are particularly appealing since many applications often deal with dividing a small number of resources among few agents (e.g. assigning meeting rooms). Turning our attention to the case when the contiguity requirement is dropped, we then show that envy-free feasible allocations exist for any number $n$ of agents and any number $m$ of layers with $m \leq n$, under mild conditions on agents' preferences.
We further show that proportional complete allocations that are both feasible and contiguous exist when the number of layers isa power of two. 
Subsequently, we show that although this result cannot be immediately extended to any number of agents and layers, a proportional complete allocation that is feasible exists when the number of layers is at most the number of agents, and can be computed efficiently. 

\begin{table}[t]
\small 
\centering
\begin{tabular}{@{}lllll@{}}
\toprule
Agents ($n$)            & Layers ($m$)                         & EF & Prop &  \\ \midrule
2               & 2                         &  \cmark (Thm. \ref{thm:EF:two} ) &   \cmark (Thm. \ref{thm:exponential})    &  \\
3               & 2                         &  \cmark (Thm. \ref{thm:EF:any}$^\diamondsuit\dagger$) &    \cmark (Thm. \ref{thm:exponential})  &  \\
any $n\geq m$              & $2^{a}$,   \scriptsize{$a\in\bbZ_{+}$}   &   \cmark (Thm. \ref{thm:EF:any}$^\diamondsuit\dagger$)  &   \cmark (Thm. \ref{thm:exponential})   &  \\ 
any $n\geq m$  & any $m$ &  \cmark (Thm. \ref{thm:EF:any}$^\diamondsuit\dagger$)  &   \cmark (Thm. \ref{thm:prop:feasible:any}$^\diamondsuit$)   &  \\ \bottomrule
\end{tabular}
\caption{The overview of our results. $\dagger$ assumes continuity of value density functions. $\diamondsuit$ indicates that existence holds without contiguity requirement. Note that when $m> n$, no complete and feasible (non-overlapping) solution exists.}
\end{table}

\subsection{Related Work}
In recent years, cake cutting has received significant attention in artificial intelligence and economics as a metaphor for algorithmic approaches in achieving fairness in allocation of resources \citep{procaccia2013cake,branzei2019communication,kurokawa2013cut,aziz2016discrete}. 
Recent studies have focused on the fair division of resources when agents have requirements over multiple resources that must be simultaneously allocated in order to carry out certain tasks (e.g. CPU and RAM) \citep{Ghodsi:2011:DRF:1972457.1972490,gutman2012fair,parkes2015beyond}. 
The most relevant work to ours is the envy-free multi-cake fair division that considers dividing multiple cakes among agents with linked preferences over the cakes. Here, agents can simultaneously benefit from all allocated pieces with no constraints. They show that envy-free divisions with only few cuts exist for two agents and many cakes, as well as three agents and two cakes~\citep{cloutier2010two,lebert2013envy,nyman2020fair}. In contrast, a multi-layered cake cutting requires non-overlapping pieces. Thus, \cite{cloutier2010two}'s generalized envy-freeness notion on multiple cakes does not immediately imply envy-freeness in our setting and no longer induces a feasible division.

\section{Our Model}
Our setting includes a set of {\em agents} denoted by $N=[n]$, a set of {\em layers} denoted by $L=[m]$, where for a natural number $s \in \bbN$, $[s]=\{1,2,\ldots,s\}$.
Given two real numbers $x,y \in \bbR$, we write $[x,y]=\{\, z \in \bbR \mid x \le z \le y \,\}$ to denote an interval. We denote by $\bbR_{+}$ (respectively $\bbZ_{+}$) the set of non-negative reals (respectively, integers) including $0$. 
A {\em piece} of cake is a finite set of disjoint subintervals of $[0,1]$. We say that a subinterval of $[0,1]$ is a {\em contiguous piece} of cake. An {\em $m$-layered cake} is denoted by $\calC=(C_j)_{j \in L}$ where $C_j \subseteq [0,1]$ is a contiguous piece for $j \in L$. We refer to each $j \in L$ as $j$-th {\em layer} and $C_j$ as $j$-th {\em layered cake}. 

Each agent $i$ is endowed with a non-negative {\em integrable density function} $v_{ij}:C_j \rightarrow \bbR_{+}$. For a given piece of cake $X$ of $j$-th layer, $V_{ij}(X)$ denotes the value assigned to it by agent $i$, i.e., $V_{ij}(X)=\sum_{I \in X}\int_{x \in I} v_{ij}(x) dx$. These functions are assumed to be {\em normalized} over layers: for each $i \in N$, $\sum_{j \in L}V_{ij}(C_j)=1$. A {\em layered piece} is a sequence $\calX=(X_j)_{j \in L}$ of pieces of each layer $j \in L$; a layered piece is said to be {\em contiguous} if each $X_j$ is a contiguous piece of each layer. 
We assume valuation functions are \emph{additive} on layers and write $V_{i}(\calX)=\sum_{j \in L}V_{ij}(X_{j})$.

A layered contiguous piece is said to be {\em non-overlapping} if no two pieces from different layers overlap, i.e, for any pair of distinct layers $j,j' \in L$ and for any $I \in X_j$ and $I' \in X_{j'}$, $I \cap I'=\emptyset$. For two layered pieces $\calX$ and $\calX'$, we say that agent $i$ {\em weakly prefers} $\calX$ to $\calX'$ if $V_i(\calX) \geq V_i(\calX')$. 

A {\em multi-allocation} $\calA=(\calA_1,\calA_2,\ldots,\calA_n)$ is a partition of the $m$-layered cake $\calC$ where each $\calA_i=(A_{ij})_{j \in L}$ is a layered piece of the cake allocated to agent $i$; we refer to each $\calA_i$ as a {\em bundle} of $i$. For a multi-allocation $\calA$ and $i \in N$, we write $V_{i}(\calA_i)=\sum_{j \in L}V_{ij}(A_{ij})$ to denote the value of agent $i$ for $\calA_i$. A {\em multi-allocation} $\calA$ is said to be 
\begin{itemize}
\item {\em contiguous} if each $\calA_i$ for $i \in N$ is contiguous; 
\item {\em feasible} if each $\calA_i$ for $i \in N$ is non-overlapping.
\end{itemize}

We focus on \emph{complete} multi-allocations where the entire cake must be allocated.
Notice that some layers may be disjoint (see Figure \ref{fig:server}), and the number of agents must exceed the number of layers, i.e. $n \geq m$; otherwise the multi-allocation will contain overlapping pieces. We illustrate our model in the following example. 

\begin{example}[Resource sharing]
Suppose that there are three meeting rooms $r_1$, $r_2$, and $r_3$ with different capacities, and three researchers Alice, Bob, and Charlie. The first room is available all day, the second and the third rooms are only available in the morning and late afternoon, respectively (see Fig.~\ref{fig:server}). Each researcher has a preference over the access time to the shared rooms. For example, Alice wants to have a group meeting in the larger room in the morning and then have an individual meeting in the smaller one in the afternoon. 
\end{example}

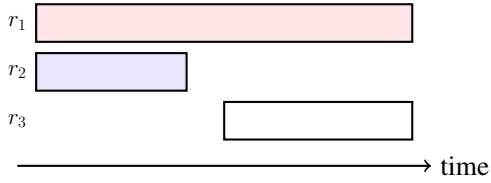
\begin{figure}[hbt]
\centering
\begin{tikzpicture}[scale=0.5, transform shape]
\draw[fill=red!10, thick] (0,1.3) rectangle (10,2.3);
\draw[fill=blue!10, thick] (0,0) rectangle (4,1); 
\draw[thick] (5,-1.3) rectangle (10,-0.3);

\node at (-0.5,1.8) {\huge $r_1$};
\node at (-0.5,0.5) {\huge $r_2$};
\node at (-0.5,-0.8) {\huge $r_3$};

\draw[thick,->] (-0.5,-2) -- (10.5,-2);
\node at (11.4,-2) {\Huge time};

\end{tikzpicture}
\caption{Example of a multi-layered cake. There are three meeting rooms $r_1$, $r_2$, and $r_3$ with different capacities, shared among several research groups.}
\label{fig:server}
\end{figure}

\paragraph{Fairness.}
A multi-allocation is said to be {\em envy-free} if no agent {\em envies} the others, i.e., $V_{i}(\calA_i) \ge V_{i}(\calA_{i'})$ for any pair of agents $i,i' \in N$. A multi-allocation is said to be {\em proportional} if each agent gets his {\em proportional fair share}, i.e., $V_{i}(\calA_i) \ge \frac{1}{n}$ for any $i \in N$. The following implication, which is well-known for the standard setting, holds in our setting as well.

\begin{lemma}\label{lem:propEF}
An envy-free complete multi-allocation satisfies proportionality. 
\end{lemma}
\begin{proof}
Consider an envy-free complete multi-allocation $\calA_i=(A_{ij})_{j \in L}$ and an agent $i \in N$. By envy-freeness, we have that $V_{i}(\calA_i) \geq V_{i}(\calA_j)$ for any $j \in N$. Summing over $j \in N$, we get $V_{i}(\calA_i) \geq \frac{1}{n}\sum_{j \in N}V_{i}(\calA_j)=\frac{1}{n}$ by additivity. 
\end{proof}

\paragraph{The $m$-layered cuts.}
In order to cut the layered cake while satisfying the non-overlapping constraint, we define a particular approach for partitioning the entire cake into diagonal pieces. Consider the $m$-layered cake $\calC$ where $m$ is an even number. For each point $x$ of the interval $[0,1]$, we define
\begin{itemize}
\item $LR(x,\calC)=(\bigcup^{\frac{m}{2}}_{j=1}C_j \cap [0,x]) \cup (\bigcup^{m}_{j=\frac{m}{2}+1}C_j \cap [x,1])$; 
\item $RL(x,\calC)=(\bigcup^{\frac{m}{2}}_{j=1}C_j \cap [x,1]) \cup (\bigcup^{m}_{j=\frac{m}{2}+1}C_j \cap [0,x])$. 
\end{itemize}
$LR(x,\calC)$ consists of the top-half subintervals of points left of $x$ and the lower-half subintervals of points right of $x$; similarly, $RL(x,\calC)$ consists of the top-half subintervals of points right of $x$ and the lower-half subintervals of points left of $x$ (Fig. \ref{fig:LR:RL}). We abuse the notation and write $LR(x)=LR(x,\calC)$ and $RL(x)=RL(x,\calC)$ if $\calC$ is clear from the context. 

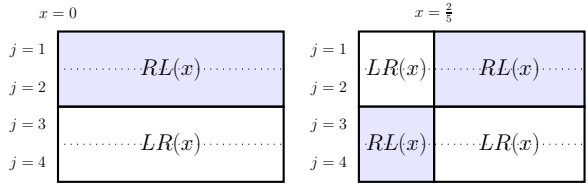
\begin{figure}[ht]
\centering
\begin{tikzpicture}[scale=0.5, transform shape]

\draw[fill=blue!10, thick] (0,0) rectangle (6,2); 
\draw[thick] (0,-2) rectangle (6,0);
\draw[dotted] (0,1) -- (6,1);
\draw[dotted] (0,-1) -- (6,-1);

\node[thick] at (0,2.5) {\Large $x=0$};
\node at (3,-1) {\huge $LR(x)$};
\node at (3,1) {\huge $RL(x)$};

\node at (-0.8,1.5) {\Large $j=1$};
\node at (-0.8,0.5) {\Large $j=2$};
\node at (-0.8,-0.5) {\Large $j=3$};
\node at (-0.8,-1.5) {\Large $j=4$};

\begin{scope}[xshift=8cm]

\draw[thick] (0,0) rectangle (2,2);
\draw[fill=blue!10, thick] (0,-2) rectangle (2,0); 
\draw[fill=blue!10, thick] (2,0) rectangle (6,2); 
\draw[thick] (2,-2) rectangle (6,0);
\draw[dotted] (0,1) -- (6,1);
\draw[dotted] (0,-1) -- (6,-1);

\node[thick] at (2.0,2.5) {\bf \Large $x=\frac{2}{5}$};
\node at (1,1) {\huge $LR(x)$};
\node at (4,-1) {\huge $LR(x)$};
\node at (4,1) {\huge $RL(x)$};
\node at (1,-1) {\huge $RL(x)$};

\node at (-0.8,1.5) {\Large $j=1$};
\node at (-0.8,0.5) {\Large $j=2$};
\node at (-0.8,-0.5) {\Large $j=3$};
\node at (-0.8,-1.5) {\Large $j=4$};

\end{scope}
 
\end{tikzpicture}
\caption{Examples of the partitions induced by $x=0$ and $x=\frac{2}{5}$ for a {\bf four-layered} cake.}
\label{fig:LR:RL}
\end{figure}

\paragraph{Computational model.}
Following the standard {\em Robertson-Webb Model} \citep{Robertson98}, we introduce two types of queries: those for a cake on each layer (called a {\em short knife}) and those for the entire cake (called a {\em long knife}). 

\paragraph{Short knife.} Short eval query: given an interval $[x,y]$ of the $j$-th layered cake $C_j$, $eval_j(i,x,y)$ asks agent $i$ for its value $[x,y]$, i.e., $V_{ij}([x,y])$.
 Short cut query: given a point $x$ and $r \in [0,1]$, $cut_j(i,x,r)$ asks agent $i$ for the minimum point $y$ such that $V_{ij}([x,y])=r$.

\paragraph{Long knife.} Long eval query: given a point $x$, $eval(i,x)$ asks agent $i$ for its value $LR(x)$, i.e., $V_{i}(LR(x))$. 
Long cut query: given $r \in [0,1]$, $cut(i,r)$ asks agent $i$ for the minimum point $x$ such that $V_{i}(LR(x))=r$ if such point $x$ exists.

\section{Existence of a switching point}
We start by showing the existence of a point $x$ that equally divides the entire cake into two pairs of diagonal pieces, both for the individuals and for the majority; these will serve as a fundamental property in our problem. 
We say that $x \in [0,1]$ is a {\em switching point} for agent $i$ over a layered cake $\calC$ if $V_i(LR(x))=V_i(RL(x))$. 

\begin{lemma}\label{lem:switching}
Suppose that the number $m$ of layers is even. Take any $i \in N$. Let $r\in \bbR$ be such that $($i$)$ $V_i(LR(0)) \geq r$ and $V_i(RL(0)) \leq r$, or $($ii$)$ $V_i(LR(0)) \leq r$ and $V_i(RL(0)) \geq r$. 
Then, there exists a point $x \in [0,1]$ such that $i$ values $LR(x)$ exactly at $r$, i.e. $V_i(LR(x)) = r$. In particular, a switching point for $i$ always exists. 
\end{lemma}
\begin{proof}
Suppose without loss of generality that $V_i(LR(0)) \geq r$ and $V_i(RL(0)) \leq r$. Consider the function $f(x)=V_i(LR(x))$ for $x \in [0,1]$. Recall that $f(x)$ is a continuous function written as the sum of continuous functions:
$
f(x)=\sum^{\frac{m}{2}}_{j=1} V_{ij}(C_j \cap [0,x])+ \sum^{m}_{j=\frac{m}{2}+1} V_{ij}(C_j \cap [x,1]). 
$
Since $f(0) \geq r$ and $f(1) \leq r$, there is a point $x \in [0,1]$ with $f(x)=r$ by the intermediate value theorem, which proves the claim. Further, by taking $r=\frac{1}{2}$, the point $x$ where $V_i(LR(x))=\frac{1}{2}$ is a switching point for agent $i$. 
\end{proof}

We will generalize the notion of a switching point from the individual level to the majority. For layered contiguous pieces $\calI$ and $\calI'$, we say that the majority weakly prefer $\calI$ to $\calI'$ (denoted by $\calI \msucceq \calI'$) if there exists $S \subseteq N$ such that $|S| \geq \ceil{\frac{n}{2}}$ and each $i \in S$ weakly prefers $\calI$ to $\calI'$. We say that $x \in [0,1]$ is a {\em majority switching point} over $\calC$ if $LR(x) \msucceq RL(x)$ and $RL(x) \msucceq LR(x)$.
The following lemma guarantees the existence of a majority switching point, for any even number of layers and any number of agents. 

\begin{lemma}\label{lem:majority:switching}
Suppose that the number of layers, $m$, is even. Then, there exists a majority switching point for any number $n \geq m$ of agents. 
\end{lemma}
\begin{proof}
Suppose without loss of generality that the majority of agents weakly prefer $LR(0)$ to $RL(0)$. Since $LR(0)=RL(1)$ and $RL(0)=LR(1)$, this means that by the time when the long knife reaches the right-most point, i.e., $x=1$, the majority preference switches. 

Formally, consider the following set of points $x \in [0,1]$ where the majority weakly prefer $LR(x)$ to $RL(x)$:  
\[
M:= \{\, x \in [0,1] \mid LR(x) \msucceq RL(x)\,\}. 
\]
We will first show that $M$ is a compact set. Clearly, $M$ is bounded. To show that $M$ is closed, consider an infinite sequence as follows $X=\{x_k\}_{k=1,2, \ldots} \subseteq M$ that converges to $x^*$. For each $k=1,2,\ldots$, we denote by $S_k$ the set of agents who weakly prefer $LR(x_k)$ to $RL(x_k)$; by definition, $|S_k| \geq \ceil{\frac{n}{2}}$. Since there are finitely many subsets of agents, there is one subset $S_k \subseteqq N$ that appears infinitely often; let $S^*$ be such subset and $\{x^*_k\}_{k=1,2, \ldots}$ be an infinite sub-sequence of $X$ such that for each $k$, each agent in $S^*$ weakly prefers $LR(x^*_k)$ to $RL(x^*_k)$. Since the valuations $V_i$ for $i \in S^*$ are continuous, each agent $i \in S^*$ weakly prefers $LR(x^*)$ to $RL(x^*)$ at the limit $x^*$, which implies that $x^* \in M$  and hence $M$ is closed. Now since $M$ is a compact set, the supremum $t^*=\sup M$  belongs to $M$. By the maximality of $t^*$, at least $\ceil{\frac{n}{2}}$ agents weakly prefer $RL(t^*)$ to $LR(t^*)$. Since $t^* \in M$, at least $\ceil{\frac{n}{2}}$ agents weakly prefer $LR(t^*)$ to $RL(t^*)$ as well. Thus, $t^*$ corresponds to a majority switching point. 
\end{proof}

\section{Envy-free multi-layered cake cutting}
First, we will look into the problem of obtaining complete envy-free multi-allocations, while satisfying non-overlapping constraints. When there is only one layer, it is known that an envy-free contiguous allocation exists for any number of agents under mild assumptions on agents' preferences \citep{Stromquist1980,Su1999}.
Given the contiguity and feasibility constraints, the question is whether it is possible to guarantee an envy-free division in the multi-layered cake-cutting model. 

\subsection{Two agents and two layers}
We answer the above question positively for a simple, yet important, case of two agents and two layers. The standard protocol that achieves envy-freeness for two agents is known as the {\em cut-and-choose} protocol: Alice divides the entire cake into two pieces of equal value. Bob selects his preferred piece over the two pieces, leaving the remainder for Alice. 

We extend this protocol to the multi-layered cake cutting using the notion of a switching point. Alice first divides the layered cake into two {\em diagonal pieces}: one that includes the top left and lower right parts and another that includes the top right and lower left parts of the cake. Our version of the cut-and-choose protocol is specified as follows:

\vspace{5pt}

\noindent\fbox{%
	\parbox{0.985\linewidth}
	   {%
		\textbf{Cut-and-choose protocol for $n=2$ agents} over a two-layered cake $\calC$: \\
		\textit{Step 1.} Alice selects her switching point $x$ over $\calC$.\\
		\textit{Step 2.} Bob chooses a weakly preferred layered contiguous piece among $LR(x)$ and $RL(x)$. \\
		\textit{Step 3.} Alice receives the remaining piece.
}%
}
\begin{figure}[htb]
\centering
\begin{tikzpicture}[scale=0.7, transform shape]
\draw[thick] (0,0) rectangle (3,1);
\draw[fill=blue!10, thick] (3,0) rectangle (10,1); 
\draw[fill=blue!10, thick] (0,-1) rectangle (3,0); 
\draw[thick] (3,-1) rectangle (10,0); 
\node at (3.0,1.3) {\large $x$};
\node at (1.5,0.5) {\large $LR(x)$};
\node at (6.5,-0.5) {\large $LR(x)$};
\node at (1.5,-0.5) {\large $RL(x)$};
\node at (6.5,0.5) {\large $RL(x)$};
\end{tikzpicture}
\caption{Cut-and-Choose for two-layered cake}
\label{fig:EF:two}
\end{figure}

\begin{theorem}\label{thm:EF:two}
The cut-and-choose protocol yields a complete envy-free multi-allocation that is feasible and contiguous for two agents and a two-layered cake using $O(1)$ number of long eval and cut queries.
\end{theorem}
\begin{proof}
It is immediate to see that the protocol returns a complete multi-allocation where each agent is assigned to a non-overlapping layered contiguous piece. The resulting allocation satisfies envy-freeness: Bob does not envy Alice since he chooses a preferred piece among $LR(x)$ and $RL(x)$. Alice does not envy Bob by the definition of a switching point. 
\end{proof}

As we noted in Section $2$, the existence result for two agents does not extend beyond two layers: if there are at least three layers, there is no feasible multi-allocation that completely allocates the cake to two agents. 

\subsection{Three agents and two layers}
We move on to the case of three agents and two layers. We will design a variant of Stromquist's protocol that achieves envy-freeness for one-layered cake \citep{Stromquist1980}: The referee moves two knives: a short knife and a long knife. The short knife points to the point $y$ and moves from left to right over the top layer, gradually increasing the left-most top piece (denoted by $Y$). The long knife keeps pointing to the point $x$, which can partition the remaining cake, denoted by $\calC^{-y}$, into two diagonal pieces $LR(x)$ and $RL(x)$ in an envy-free manner. Each agent shouts when the left-most top piece $Y$ becomes at least as highly valuable as the preferred piece among $LR(x)$ and $RL(x)$. Some agent, say $s$, shouts eventually (before the left-most top piece becomes the top layer), assuming that there is at least one agent who weakly prefers the top layer to the bottom layer. We note that $x$ may be positioned left to $y$; see Figure \ref{fig:EF:three} for some possibilities of the long knife's locations. 

We will show that the above protocol works, for a special case when there are at most two types of preferences: In such cases, the majority switching points coincide with the switching points of an agent with the majority preference.

\begin{lemma}\label{lem:switching:identical}
Suppose that $m=2$, $n=3$, and there are two different agents $i,j \in N$ with the same valuation $V$. Then, $x$ is a majority switching point over $\calC$ if and only if $x$ is a switching point for $i$. 
\end{lemma}
\begin{proof}
Suppose that agents $i,j \in N$ have the same valuation $V$. Suppose that $x$ is a majority switching point over $\calC$. Then, at least two agents weakly prefer $LR(x)$ to $RL(x)$, meaning that at least one of the two agents $i$ and $j$ weakly prefers $LR(x)$ to $RL(x)$, which means that both agents weakly prefers $LR(x)$ to $RL(x)$ since $i$ and $j$'s valuations are identical. Similarly, both $i$ and $j$ weakly prefer $RL(x)$ to $LR(x)$. Thus, $x$ is a switching point for $i$. 
The converse direction is immediate. 
\end{proof}

An implication of the above lemma is that when performing Stromquist's protocol, one can point out to a switching point of an individual, instead of a majority one. This allows the value of each piece to change continuously. For a given two-layered cake $\calC$, we write $\calC^{-y}=(C^{-y}_1,C_2)$ as a two-layered cake obtained from $\calC$ where the first segment $[0,y]$ of the top layer is removed, i.e., $C^{-y}_1 = C_1 \setminus [0,y]$. For each majority switching point $x$ over $\calC^{-y}$, we select three different agents $\ell(x)$, $m(x)$, and $r(x)$ as follows: 
\begin{itemize}
\item $\ell(x)$ is an agent who weakly prefers $LR(x,\calC^{-y})$ to $RL(x,\calC^{-y})$;
\item $m(x)$ is an agent who is indifferent between $LR(x,\calC^{-y})$ and $RL(x,\calC^{-y})$; and 
\item $r(x)$ and agent who weakly prefers $RL(x,\calC^{-y})$ to $LR(x,\calC^{-y})$. 
\end{itemize}

\begin{theorem}\label{thm:twolayers3agents}
Suppose that $m=2$ and $n=3$. If there are two different agents with the same valuation, an envy-free complete multi-allocation that is feasible and contiguous exists. 
\end{theorem}
\begin{proof}
Assume w.l.o.g. that at least one agent prefers the top layer over the bottom layer. This means that such agent weakly prefers the top layer to any of the pieces $LR(z,\calC^{-y})$ and $RL(z,\calC^{-y})$ when $y=1$. Suppose that $i \in N$ is one of the two different agents with the same valuations. We design the following protocol for three agents over a two-layered cake: 

\vspace{3pt}
\noindent\fbox{%
	\parbox{0.985\linewidth}{%
		\textbf{Moving-knife protocol for $n=3$ agents} over a two-layered cake $\calC$: w.l.o.g. assume that at least one agent weakly prefers the top layer $(j=1)$ over the bottom layer $(j=2)$\\
		\textit{Step 1.} The referee continuously moves a short knife from the left-most point $(y=0)$ to the right-most point $(y=1)$ over the top layer, while continuously moving a long knife pointing to a switching point over $\calC^{-y}$ for $i$. 
        Let $y$ be the position of the short knife and $Y$ be the top layer piece to its left. Let $x$ be the position of the long knife.\\		
		\textit{Step 2.} The referee stops moving the short knife when some agent $s$ {\em shouts}, i.e., $Y$ becomes at least as highly valuable as the preferred piece among $LR(x,\calC^{-y})$ and $RL(x,\calC^{-y})$. \\
		\textit{Step 3.} We allocate the shouter $s$ to the left-most top piece $Y$ and partitions the rest into $LR(x,\calC^{-y})$ and $RL(x,\calC^{-y})$. 
\begin{itemize}
\item If $s=\ell(x)$, then we allocate $LR(x,\calC^{-y})$ to $m(x)$ and $RL(x,\calC^{-y})$ to $r(x)$. 
\item If $s=m(x)$, then we allocate $LR(x,\calC^{-y})$ to $\ell(x)$ and $RL(x,\calC^{-y})$ to $r(x)$. 
\item If $s=r(x)$, then we allocate $LR(x,\calC^{-y})$ to $\ell(x)$ and $RL(x,\calC^{-y})$ to $m(x)$. 
\end{itemize}
	}%
}
\vspace{3pt}

By our assumption, some agent eventually shouts and thus the protocol returns an allocation $\calA$. Clearly, $\calA$ is feasible, contiguous, and complete. Also, it is easy to see that the shouter $s$ who receives a bundle $Y$ does not envy the other two agents. The agents $i \neq s$ do not envy $s$ because the referee continuously moves both a short and a long knife. Finally, the agents $i \neq s$ do not envy each other by the definition of a majority switching point and by Lemma \ref{lem:switching:identical}. 
\end{proof}

\begin{figure}[t]
\centering
\begin{tikzpicture}[scale=0.55, transform shape]
\draw[fill=red!10, thick] (0,0) rectangle (4,1);
\draw[fill=blue!10, thick] (0,-1) rectangle (3,0); 
\draw[fill=blue!10, thick] (4,0) rectangle (10,1); 
\draw[thick] (3,-1) rectangle (10,0);

\node at (2.0,0.5) {$Y$};
\node at (1.5,-0.5) {$RL(x)$};
\node at (6.5,0.5) {$RL(x)$};
\node at (6.0,-0.5) {$LR(x)$};

\node at (3.0,1.3) {$x$};
\node at (4.0,1.3) {$y$};

\begin{scope}[yshift=-3cm]
\draw[fill=red!10, thick] (0,0) rectangle (4,1);
\draw[thick] (4,0) rectangle (10,1);
\draw[fill=blue!10, thick] (6,0) rectangle (10,1); 
\draw[fill=blue!10, thick] (0,-1) rectangle (6,0); 
\draw[thick] (6,-1) rectangle (10,0);

\node at (2.0,0.5) {$Y$};
\node at (6.0,1.3) {$x$};
\node at (4.0,1.3) {$y$};

\node at (3.0,-0.5) {$RL(x)$};
\node at (8.0,0.5) {$RL(x)$};
\node at (5.0,0.5) {$LR(x)$};
\node at (8.0,-0.5) {$LR(x)$};
\end{scope}
 
\end{tikzpicture}
\caption{Moving knife protocol for three agents over a two-layered cake. Note that the position of $x$ may appear before $y$. 
}
\label{fig:EF:three}
\end{figure}
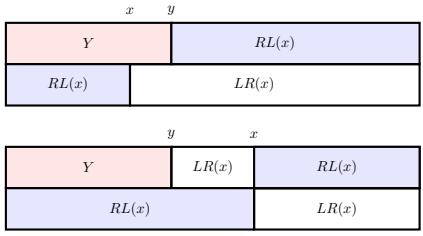

In the general case, the existence question of contiguous and feasible envy-free multi-allocations deems to be challenging due to the non-monotonicity of valuations over diagonal pieces.\footnote{See Section \ref{sec:discussion} for an extensive discussion.} 
In the next subsection, we thus turn our attention to the case when the contiguity requirement is relaxed.

\subsection{Non-connected pieces}
Having seen that an envy-free multi-allocation that is both feasible and contiguous exists for a special case, we will consider the case when the contiguity requirement is dropped, namely, agents may receive a collection of sub-intervals of each layer. 
We will show the existence of an envy-free multi-allocation that is feasible and uses at most poly$(n)$ number of cuts within each layer, assuming that each density function $v_{ij}$ is continuous. In what follows, we will reduce the problem to finding a `perfect' allocation of a one-layered cake. An allocation of a single-layered cake is called {\em perfect} if each agent values every allocated piece exactly at his proportional fair share $\frac{1}{n}$. It is known that such allocation consisting of at most poly$(n)$ number of contiguous pieces exists whenever agents' value density functions are continuous \citep{Alon1987}. It is not surprising that the existence of a perfect allocation implies the existence of an envy-free allocation over a single-layered cake. We show that this result also implies the existence of envy-free allocations over a multi-layered cake. 

\begin{theorem}\label{thm:EF:any}
Suppose that $m \leq n$ and each $v_{i,j}$ for $i \in N$ and $j \in L$ is continuous. Then, an envy-free complete feasible multi-allocation $\calA=(\calA_i)_{i \in N}$ where each piece $A_{ij}$ for agent $i \in N$ and layer $j \in L$ contains at most poly$(n)$ number of contiguous pieces exists. 
\end{theorem}
\begin{proof}
We assume without loss of generality that each layer $C_j$ is the whole interval $[0,1]$ by just putting zero valuations on the part outside $C_j$. 
Now, we construct an instance of a single-layered cake $I=[0,1]$ as follows. First, create a dummy agent $i_j$ for each agent $i \in N$ and each layer $j \in L$. Each dummy agent $i_j$ has a valuation $v'_{i_j}$ determined by agent $i$'s valuation for the $j$-th layered cake, i.e., for each sub-interval $X \subseteq I$, $v'_{i_j}(X)=v_{i,j}(X)$. 
We will show that a perfect allocation of the artificial cake among the $mn$ dummy agents induces an envy-free multi-allocation of the original layered cake. 

Specifically, take a perfect allocation $(X_1,X_2,\ldots,X_{mn})$ of this instance where each $X_t$ for $t \in [mn]$ contains at most poly$(n)$ number of contiguous pieces, which is guaranteed to exist \citep{Alon1987}. We then group each consecutive $m$ sequence of pieces together: namely, let $Y_{h}=\bigcup^{im}_{t=(i-1)m+1}X_t$ for each $h \in [n]$. By the definition of a perfect allocation, we have
\begin{align}\label{eq}
&v_{ij}(Y_h)= \frac{v_{ij}(C_j)}{n},
\end{align}
for any $h \in [n]$. Now, we partition each layer into $n$ pieces using the partition $(Y_1,Y_2,\ldots,Y_n)$ of the artificial cake and allocate to the agents so that each agent  receives exactly one piece $Y_h$ for each layer. Formally, consider a permutation $\sigma_j:[n] \rightarrow [n]$ where
\[
\sigma_j(i)
=i+j -1 \pmod{m}. 
\]
Construct an multi-allocation $\calA=(\calA_i)_{i \in N}$ where each agent $i \in N$ is assigned to $A_{ij}=Y_{\sigma_j(i)}$ for each layer $j \in L$. By our construction, each $A_{ij}$ contains at most poly$(n)$ number of contiguous pieces. Also, each layered piece $\calA_i$ is non-overlapping as $(Y_1,Y_2,\ldots,Y_n)$ is a partition of the interval $[0,1]$. By \eqref{eq}, it is immediate to see that $\calA$ is envy-free. 
\end{proof}

\section{Proportional multi-layered cake cutting}
Focusing on a less demanding fairness notion, it turns out that a complete proportional multi-allocation that is both feasible and contiguous exists, for a wider class of instances, i.e., when the number $m$ of layers is some power of two, and the number $n$ of agents is at least $m$. Notably, we show that the problem can be decomposed into smaller instances when the number of agents is at least the number of layers and the number of layers is a power of two. Building up on the {\em base case} of two layers, our algorithm recursively calls the same algorithm to decide on how to allocate the cake of the sub-problems. We further show that if we relax the contiguity requirement, a proportional feasible multi-allocation can be computed efficiently whenever $m \leq n$. 

\subsection{Connected pieces}
In this subsection, we will show that a proportional complete multi-allocation exists for any $n \geq m$ when $m$ is some power of $2$. We start by presenting two auxiliary lemmata. We define a {\em merge} of two disjoint contiguous pieces $I_j$ and $I_{j'}$ of layers $j$ and $j'$ as replacing the $j$-th layered cake with the union $I_j \cup I_{j'}$ and removing $j'$-th layered cake. The {\em merge} of a finite sequence of mutually disjoint contiguous pieces $(I_1,\ldots,I_k)$ can be defined inductively: merge $(I_1,\ldots,I_{k-1})$ and then apply the merge operation to the resulting outcome and $I_k$. 
Now we observe that if there are two disjoint layers, one can safely merge these layers and reduce the problem size. 

\begin{lemma}\label{lem:merge}
Suppose that $C_j$ and $C_{j'}$ are two disjoint layers of a layered cake $\calC$, and $\calC'$ is obtained from $\calC$ by merging $C_j$ and $C_{j'}$. Then, each non-overlapping contiguous layered piece of $\calC'$ is a non-overlapping contiguous layered piece of the original cake $\calC$. 
\end{lemma}
\begin{proof}
Suppose that $C_j$ and $C_{j'}$ are two disjoint layers of a layered cake $\calC=(C_{t})_{t \in L}$, and the layered cake $\calC'=(C'_{t})_{t \in L \setminus \{j'\}}$ is obtained from $\calC$ by merging $C_j$ and $C_{j'}$. Let $\calX'=(X'_{t})_{t \in L \setminus \{j'\}}$ be a non-overlapping contiguous piece of $\calC'$. Consider the corresponding layered piece $\calX=(X_{t})_{t \in L}$ of the original cake $\calC$ where $X_t=X'_{t}$ for $t \in L \setminus \{j,j'\}$ and $X_t = X'_t \cap C_{t}$ for $t \in \{j,j'\}$. It is immediate to see that $\calX'$ is non-overlapping and contiguous, since $C_j$ and $C_{j'}$ are disjoint. 
\end{proof}

The above lemma can be generalized further: Let $\calC$ be a $2m$-layered cake and $x \in [0,1]$. We define a {\em merge} of $LR(x)=(S_j)_{j \in L}$ by merging the pair $(S_j,S_{j+m})$ for each $j \in [m]$. A {\em merge} of $RL(x)$ can be defined analogously. Such operation still preserves both feasibility and contiguity. 

\begin{corollary}\label{cor:merge}
Let $\calC$ be a $2m$-layered cake and $x \in [0,1]$. Suppose that $\calC'$ is a $m$-layered cake obtained by merging $LR(x,\calC)$ or $RL(x,\calC)$. Then, each non-overlapping contiguous layered piece of $\calC'$ is a non-overlapping contiguous layered piece of the original cake $\calC$. 
\end{corollary}
\begin{proof}
Suppose that $\calC'$ is a $m$-layered cake obtained by merging $LR(x,\calC)=(S_j)_{j \in L}$ of a $2m$-layered cake $\calC$. By Lemma \ref{lem:merge}, a non-overlapping contiguous layered piece of the cake obtained from each merge of the pair $(S_j,S_{j+m})$ for $j \in [m]$ still corresponds to a non-overlapping and contiguous piece of the original cake. Thus, the claim holds. An analogous argument applies to the case when we merge $RL(x,\calC)$. 
\end{proof}

We are now ready to prove that a proportional complete multi-allocation exists for any $n=m$ when $m$ is some power of $2$. In essence, the existence of a majority switching point, as proved in Lemma \ref{lem:majority:switching}, allows us to divide the problem into two instances. We will repeat this procedure until the number of layers of the subproblem becomes $2$, for which we know the existence of a proportional, feasible, contiguous multi-allocation by Theorem \ref{thm:EF:two}.  

\begin{theorem}\label{thm:prop:base}
A proportional complete multi-allocation that is feasible and contiguous exists, for any number $m$ of layers and any number $n= m$ of agents where $m=2^a$ for some $a \in \bbZ_{+}$. 
\end{theorem}
\begin{proof}
We design the following recursive algorithm $\calD$ that takes a subset $N'$ of agents with $|N'| \geq 2$, a $|L'|$-layered cake $\calC'$, and a valuation profile $(V_{i})_{i \in N'}$, and returns a proportional complete multi-allocation of the cake to the agents which is feasible. Suppose that $m=n$. If $m=n=1$, then we allocate the entire cake to the single agent. If $m=n=2$, we run the cut-and-choose algorithm as described in the proof of Theorem \ref{thm:EF:two}. 
Now consider the case when $m=n=2^a$ for some integers $a \geq 1$. Then the algorithm finds a majority switching point $x$ over $\calC'$. We let $\calI_1=LR(x)$ and $\calI_2=RL(x)$. By definition of a majority switching point and the fact that $n$ is even, we can partition the set of agents $N'$ into $N_1$ and $N_2$ where $N_1$ is the set of agents who weakly prefer $\calI_1$ to $\calI_2$, $N_2$ be the set of agents who weakly prefer $\calI_2$ to $\calI_1$, and $|N_k| = \frac{|N'|}{2}$ for each $k=1,2$. 
We apply $\calD$ to the merge of $\calI_k$ with the agent set $N_k$ for each $k=1,2$, respectively. 

We will show that by induction on the exponential $a$, that the complete multi-allocation $\calA$ returned by $\calD$ satisfies proportionality as well as feasibility and contiguity. This is clearly true when $m=n=2$ due to Lemma \ref{lem:propEF} and Theorem \ref{thm:EF:two}. Suppose that the claim holds for $m=n=2^a$ with $1 \leq a \leq k-1$; we will prove it for $a=k$. Suppose that the algorithm divides the input cake $\calC'$ via a majority switching point $x$ into $\calI_1=LR(x)$ and $\calI_2=RL(x)$. Suppose that $(N_1,N_2)$ is a partition of the agents where $N_1$ is the set of agents who weakly prefer $\calI_1$ to $\calI_2$, $N_2$ is the set of agents who weakly prefer $\calI_2$ to $\calI_1$, and $|N_k| = \frac{|N'|}{2}$ for each $k=1,2$. Observe that each agent $i \in N_1$ weakly prefers $\calI_1$ to $\calI_2$ and thus $V_i(\calI_1) \geq \frac{1}{2}V_i(\calC')$. Similarly, $V_i(\calI_2) \geq \frac{1}{2}V_i(\calC')$ for each $i \in N_2$. Thus, by the induction hypothesis, each agent $i$ has value at least $\frac{1}{|N'|}V_i(\calC')$ for its allocated piece $\calA_{i}$. Further, by Corollary \ref{cor:merge}, each non-overlapping contiguous layered piece of the merge of $\calI_1$ (respectively, $\calI_2$) is a contiguous non-overlapping layered piece of the original cake $\calC$. By the induction hypothesis, the algorithm outputs a multi-allocation of each merge that is contiguous. Thus, the algorithm returns a proportional complete multi-allocation that is feasible and contiguous.
\end{proof}

We will generalize the above theorems to the case when the number of agents is strictly greater than the number of layers. Intuitively, when $n>m$, then there is at least one layer whose sub-piece can be `safely' allocated to some agent without violating the non-overlapping constraint. 

\begin{theorem}\label{thm:exponential}
A proportional complete multi-allocation that is feasible and contiguous exists, for any number $m$ of layers and any number $n \geq m$ of agents where $m=2^a$ for some $a \in \bbZ_{+}$. 
\end{theorem}
\begin{proof}
We design the following recursive algorithm $\calD$ that takes a subset $N'$ of agents with $|N'| \geq 2$, a $|L'|$-layered cake $\calC'$, and a valuation profile $(V_{i})_{i \in N'}$, and returns a proportional complete multi-allocation of the layered cake to the agents  which is feasible. For $n=m$, we apply the algorithm described in the proof of Theorem \ref{thm:prop:base}. Suppose that $n>m$. The algorithm first identifies a layer $C_j$ whose entire valuation is at least $\frac{1}{n}$ for some agent; assume w.l.o.g. that $j=1$. We move a knife from left to right over the top cake $C_1$ until some agent $i$ {\em shouts}, i.e., agent $i$ finds the left contiguous piece $Y$ at least as highly valued as his proportional fair share $\frac{1}{n}$. The algorithm $\calD$ then gives the piece to the shouter. To decide on the allocation of the remaining items, we apply $\calD$ to the reduced instance $(N' \setminus \{i\},(C'_j)_{j \in L},(V_{i'})_{i' \in N' \setminus \{i\}})$ where $C'_j=C_j \setminus Y$ for $j=1$ and $C'_j=C_j$ for $j \neq 1$. 

We will prove by induction on $|N'|$ that the complete multi-allocation $\calA=(\calA_1,\calA_2,\ldots,\calA_n)$ returned by $\calD$ satisfies proportionality as well as feasibility and contiguity. This is clearly true when $m=|N'|$, due to Theorem  \ref{thm:prop:base}. Suppose that the claim holds for $|N'|$ with $m \leq |N'| \leq k-1$; we will prove it for $|N'|=k$. Suppose agent $i$ is the shouter who gets the left contiguous piece $Y$. Clearly, agent $i$ receives her proportional share under $\calA$. Observe that all remaining agents have the value at least $\frac{|N'|-1}{|N'|}V_i(\calC')$ for the remaining cake. Thus, by the induction hypothesis, each agent $i' \neq i$ has value at least $\frac{1}{|N'|}V_i(\calC')$ for its allocated piece $\calA_{i'}$. The feasibility and contiguity of $\calA$ are immediate by the induction hypothesis. This completes the proof. 
\end{proof}

\subsection{Non-connected pieces}
Since envy-freeness implies proportionality, Theorem \ref{thm:EF:any} in the previous section implies the existence of a proportional feasible multi-allocation when agents' value density functions are continuous. We strengthen this result, by showing that such desirable allocation exists for a more general case and by providing an efficient algorithm for finding one. 

\begin{theorem}\label{thm:prop:feasible}
A proportional complete multi-allocation that is feasible exists when $m =n$ and can be computed using $O(nm^2)$ number of short eval queries and $O(nm)$ number of long eval and cut queries. Further, each bundle of the resulting multi-allocation includes at most two contiguous pieces within each layer. 
\end{theorem}

Below, we show that each agent can divide the entire cake into $n$ equally valued layered pieces. A multi-allocation $\calA$ is {\em equitable} if for each agent $i \in N$, $V_{i}(\calA_i)=\frac{1}{n}$. We design a recursive algorithm that iteratively finds two layers for which one has value at most $\frac{1}{m}$ and at least $\frac{1}{m}$ and removes a pair of diagonal pieces of value exactly $\frac{1}{m}$ from the two layers.

\begin{lemma}\label{lem:equitable}
For any number $m$ of layers and any number $n = m$ of agents with the identical valuations, an equitable complete multi-allocation that is feasible and contiguous exists and can be found using $O(nm^2)$ number of short eval queries and $O(nm)$ number of long cut queries.  
\end{lemma}
\begin{proof}
We denote by $V=V_i$ the valuation function for each agent $i \in N$. 
Consider the following recursive algorithm $\calD$ that takes a subset $N'$ of agents with $|N'| \geq 1$, a $|L'|$-layered cake $\calC'$, and a valuation profile $(V_{i})_{i \in N'}$, and returns an equitable complete multi-allocation of the layered cake to the agents. When $|L'|=|N'|=1$, then the algorithm allocates the entire cake to the single agent. 
Suppose that $|L'|=|N'| \geq 2$. 
The algorithm first finds a layer $j$ whose entire value is at most $\frac{1}{m}$ and another layer $j'$ whose entire value is at least $\frac{1}{m}$. The algorithm $\calD$ then finds a point $x \in [0,1]$ where $V(S_{j} \cup S_{j'})=\frac{1}{m}$ for $S_j=C_j \cap [0,x]$ and $S_{j'}=C_j \cap [x,1]$; such point exists due to Lemma \ref{lem:switching}. We allocate $S_{j} \cup S_{j'}$ to one agent and apply $\calD$ to the remaining cake $\calC''$ with $|N'|-1$ agents where $\calC''$ is obtained from merging the remaining $j$-th layered cake $C_j \setminus S_j$ and the $j'$-th layered cake $C_{j'} \setminus S_{j'}$. The correctness of the algorithm as well as the bound on the query complexity are immediate. 
\end{proof}

\begin{figure*}[hbt]
\centering
\begin{tikzpicture}[scale=0.6, transform shape]

\draw[thick] (0,0) rectangle (3,1);
\draw[fill=red!10, thick] (3,0) rectangle (6,1);
\draw[fill=blue!10, thick] (0,-1) rectangle (2,0); 
\draw[fill=red!10, thick] (2,-1) rectangle (3,0);
\draw[thick] (3,-1) rectangle (6,0);
\draw[fill=red!10, thick] (0,-2) rectangle (2,-1);
\draw[fill=blue!10, thick] (2,-2) rectangle (6,-1);

\node at (4.5,0.5) {$I_{11}$};
\node at (2.5,-0.5) {$I_{12}$};
\node at (1,-1.5) {$I_{13}$};

\node at (1.5,0.5) {$I_{21}$};
\node at (4.5,-0.5) {$I_{22}$};

\node at (1,-0.5) {$I_{32}$};
\node at (4,-1.5) {$I_{33}$};

\draw[ultra thick,->] (6.5,-0.5)--(7.5,-0.5);

\begin{scope}[xshift=8cm,yshift=-0.5cm]
\node[fill=red!10,draw, circle](I1) at (1,1) {$\calI_1$};
\node[draw, circle](I2) at (3,1) {$\calI_2$};
\node[fill=blue!10,draw, circle](I3) at (5,1) {$\calI_3$};

\node[fill=gray!10,draw, circle](a1) at (1,-1) {$1$};
\node[fill=gray!10,draw, circle](a2) at (3,-1) {$2$};
\node[fill=gray!10,draw, circle](a3) at (5,-1) {$3$};

\draw[-,red, >=latex,thick] (I1)--(a1);
\draw[-, >=latex,thick] (I2)--(a1);
\draw[-,>=latex,thick] (I3)--(a1);
\draw[-, >=latex,thick] (I2)--(a2);
\draw[-, >=latex,thick] (I2)--(a3);

\draw[ultra thick,->] (6.5,0)--(7.5,0);
\end{scope}

\begin{scope}[xshift=16cm,yshift=0.5cm]
\draw[thick] (0,-1) rectangle (3,0);

\draw[fill=blue!10, thick] (0,-2) rectangle (2,-1); 
\draw[thick] (3,-1) rectangle (6,0);
\draw[fill=blue!10, thick] (2,-2) rectangle (6,-1); 

\node at (1.5,-0.5) {$I_{21}$};
\node at (4.5,-0.5) {$I_{22}$};

\node at (1,-1.5) {$I_{32}$};
\node at (4,-1.5) {$I_{33}$};

\draw[thick,dotted] (4,0.5)--(4,-2.5);
\end{scope}
 
\end{tikzpicture}
\caption{Protocol for proportionality for three agents and three layers. Agent $1$ divides the entire cake into three equally valued layered pieces $\calI_1$, $\calI_2$, and $\calI_3$ (the left-most picture). Here, $\calI_i=(I_{ij})_{j=1,2,3}$ for each $i =1,2,3$ where $I_{23}=I_{31}=\emptyset$. In the middle picture, agent $1$ is adjacent to every piece, meaning that he has value at least proportional fair share for every piece; on the other hand, the other agents are adjacent to the second piece only. The maximum envy-free matching is an edge between $\calI_1$ and agent $1$ (red edge), so the algorithm allocates $\calI_1$ to agent $1$ and merges $\calI_2$, and $\calI_3$. Then it applies the cut-and-choose among the remaining agents (the right-most picture).}
\label{fig:PROP:three}
\end{figure*}
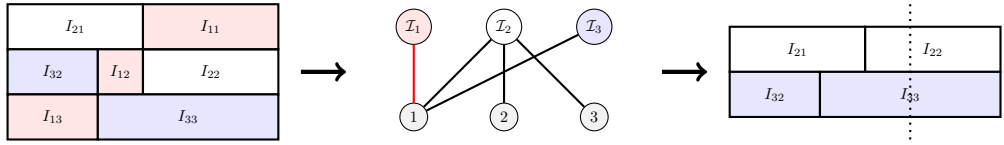

Equipped with Lemma \ref{lem:equitable}, we will prove Theorem \ref{thm:prop:feasible} by recursively computing an {\em envy-free matching} between $n$ agents and $n$ layered pieces where one agent has proportional fair share for every piece. Specifically, given a bipartite graph $G$ with one side being the set of agents and the other side being the set of items, an envy-free matching $M$ of $G$ is a matching where no unmatched agent {\em envies} some matched agent, i.e., no unmatched agent is adjacent to any matched item in $G$. The problem of finding an envy-free matching of maximum size can be solved in polynomial time \citep{EladErel,GAN2019}. Further, using Hall’s type condition, it can be easily shown that if there is one agent who is adjacent to every item and the number of agents is at most the number of items, then there is a non-empty envy-free matching (Corollary $1.4$ $($c$)$ of \citet{EladErel}).  

\begin{proof}[Proof of Thm. \ref{thm:prop:feasible}]
In order to obtain a proportional feasible multi-allocation, we will recursively compute a non-empty envy-free matching:
For each iteration, let one agent partition $n$-equally valued layered pieces, find a maximum envy-free matching between agents and pieces, and assign the matched agents to the matched pieces. By Corollary $1.4$ $($c$)$ of \citet{EladErel}, the envy-free matching computed at each step is non-empty; thus, at least one agent is matched and we will apply the same procedure to the unmatched agents and pieces. The formal description is given as follows. See Figure \ref{fig:PROP:three} for an illustration for $m=n=3$. 

\vspace{5pt}
\noindent\fbox{%
	\parbox{0.985\linewidth}{%
		\textbf{A protocol for proportional feasible multi-allocations for $n = m$ agents} over a $m$-layered cake $\calC$: \\
		\textit{Step 1.} One agent partitions the cake into $n$ non-overlapping layered contiguous pieces $\calI_1,\calI_2, \ldots, \calI_n$ which she considers of equal value, using the algorithm in the proof of Lemma \ref{lem:equitable}.\\
		\textit{Step 2.} Construct a bipartite graph $G$ with the agents being one side and the pieces being on the other side, where there is an edge between agent $i$ and $\calI_h$ if agent $i$ has value at least his proportional far share for $\calI_h$. \\
		\textit{Step 3.} Compute a maximum-size envy-free matching $M$ of $G$. Assign matched agent to the corresponding piece. \\
		\textit{Step 4.} For each unmatched piece $\calI_h$, merge all the disjoint contiguous pieces in $\calI_h$, and create a $m-\ell$-layered cake $\calC'$ consisting of each merge of $\calI_h$ where $\ell$ is the number of matched pieces. Apply the same protocol to $\calC'$ among the remaining unmatched agents. 
	}%
}
\vspace{5pt}

We will show by the induction on the number of agents $n=m$ that the resulting multi-allocation $\calA$ is proportional and feasible. The claim clearly holds for $n=m=1$. Suppose that the claim holds for $n$ with $m=n \leq k-1$; we will prove it for $m=n=k$. We will first show that the resulting multi-allocation $\calA$ is proportional. Clearly, the agents who get matched has proportional fair share for his bundle. Further, each $i$ of the remaining unmatched $n-\ell$ agents have value less than $\frac{1}{n}$ for each matched piece and thus has at least $1-\frac{\ell}{n}$ for the remaining unmatched pieces; thus, $V_i(\calA_i)$ is at least $\frac{1}{n}$. It can be easily verified that each bundle $\calA_i$ is non-overlapping. Each iteration requires $O(m^2)$ number of short eval queries and $O(m)$ number of long cut queries for the cutter, and $O(m^2)$ number of short eval queries for each agent. Further, the number of iterations is at most $n$, which proves the bound on the query complexity. 
This completes the proof.
\end{proof}

Similarly to the proof for Theorem \ref{thm:exponential}, we can generalize the above theorem to the case when the number of agents is strictly greater than the number of layers. 

\begin{theorem}\label{thm:prop:feasible:any}
A proportional complete multi-allocation that is feasible exists and can be computed using $O(nm^2)$ number of short eval queries and $O(nm)$ number of long cut queries, for any number $m$ of layers and any number $n \geq m$ of agents. 
\end{theorem}

It remains open whether a proportional contiguous multi-allocation exists when the number of layers is three. A part of the reason is that our algorithm for finding an equitable multi-allocation (Lemma \ref{lem:equitable}) may not return a `balanced' partition: The number of pieces contained in each layered piece may not be the same when the number of layers is odd. For example, one layered piece may contain pieces from three different layers while the other two parts may contain pieces from two different layers, as depicted in Figure \ref{fig:PROP:three}.

\section{Discussion} \label{sec:discussion}
We initiated the study of multi-layered cake cutting, demonstrating the rich and intriguing mathematical feature of the problem. There are several exciting questions left open for future work. Below, we list some of them. 

\begin{itemize}
\item {\bf Existence of fair allocations}:
We have seen that an envy-free contiguous and feasible multi-allocation of a two-layered cake exists for two or three agents with at most two types of preferences. An interesting open problem is whether such allocation also exists for any number of agents over a two-layered cake. One might expect that the Simmon-Su's technique \citep{Su1999} using Sperner's Lemma can be adopted to our setting by considering all possible diagonal pieces. However, this approach may not work because multi-layered cake-cutting necessarily exhibits non-monotonicity in that the value of a pair of diagonal pieces may decrease when the knife moves from left to right. For proportionality, one intriguing future direction is extending our existence result for $m=2^a$ to any $m$. This requires careful consideration of contiguity and feasibility, which are often at odds with completeness.

\item {\bf Query complexity of fair allocations}:
The query complexity of finding an envy-free feasible multi-allocation is open in the multi-layered cake-cutting problem. In particular, it would be challenging to extend the celebrated result of \citet{aziz2016discrete} -- who showed the existence of a bounded protocol for computing an envy-free allocation of a single-layered cake with any number $n$ of agents --  to our setting. We expect that a direct translation may not work, due to the intricate nature of the feasibility constraint. With respect to proportionality, our existence proof implies that if there is a way to compute a majority switching point efficiently, one can compute a proportional contiguous feasible multi-allocation for special cases when the number of layers is a power of two. It is open whether such cutting point can be computed using a bounded number of queries.

\item {\bf Approximate fairness}: 
In the presence of contiguity requirement, it is known that no finite protocol computes an envy-free allocation even for three agents and a single-layered cake \citep{Stromquist2008}. However, several positive results are known when the aim is to approximately bound the envy between agents \citep{Deng2012,Goldberg2020,Arunachaleswaran19}. 
Pursuing a similar direction in the context of multi-layered cake cutting would be an interesting research topic. 

\item {\bf Efficiency requirement}: 
Besides fairness criteria, another basic desideratum is economic efficiency. In the context of a single-layered cake cutting, several works studied the relation between fairness and efficiency \citep{CohlerLPP11,Bei12, AumannDomb2010, AumannDH13}. 
The question of what welfare guarantee can be achieved together with fairness is open in our model. In particular, it would be interesting to investigate the compatibility of the fairness notions with an efficiency requirement, under feasibility constraints. 


\end{itemize}

\section*{Acknowledgement}
Hadi Hosseini was supported by the National Science Foundation (grant IIS-1850076).
Ayumi Igarashi was supported by the KAKENHI Grant-in-Aid for JSPS Fellows no. 18J00997 and JST, ACT-X. The authors would like to thank Yuki Tamura for introducing the problem to us and for the fruitful discussions. The authors also acknowledge the helpful comments by the GAIW and IJCAI reviewers. 
We are grateful to an anonymous GAIW reviewer for the proof of the existence of an envy-free multi-allocation in the general case.

\bibliographystyle{named}
\bibliography{abb,cakeref,cake}

\end{document}